\documentclass[12pt]{article}
\usepackage{amsmath}
\usepackage{amsthm}
\usepackage[onehalfspacing]{setspace}
\usepackage{verbatim}
\usepackage{amsfonts}
\usepackage{amssymb}

\usepackage{qtree}

\usepackage{tikz}

\usepackage{color}
\definecolor{ao(english)}{rgb}{0.0, 0.5, 0.0}

\usepackage{multirow,array}

\usepackage[unicode=true,pdfusetitle,
 bookmarks=true,bookmarksnumbered=false,bookmarksopen=false,
 breaklinks=false,pdfborder={0 0 1},backref=true,colorlinks=true,citecolor=blue,linkcolor=blue]
 {hyperref}

\usepackage{latexsym}
\usepackage{float}
\usepackage{graphicx}
\usepackage{enumitem}
\usepackage[left=1.3in,right=1.3in,top=1.2in,bottom=1.2in]{geometry}
\usepackage[font=small,labelsep=none]{caption}

\usepackage{natbib}

\usepackage[title]{appendix}
\usepackage{graphicx}

\usepackage{CJK}

\usepackage{tikz}
\usetikzlibrary{calc}

\newcommand{\blacktriangleq}{\mathrel{\ooalign{$\blacktriangleright$\cr\hidewidth\raise-1.1ex\hbox{$-$}\hidewidth}}}

\newtheorem{theorem}{Theorem}

\newtheorem{corollary}[theorem]{Corollary}

\newtheorem{lemma}{Lemma}

\newtheorem{proposition}{Proposition}

\newtheorem{fact}{Fact}[section]

\newtheorem*{lemma*}{Lemma}
\newtheorem*{theorem*}{Theorem}

\newtheorem*{proposition*}{Proposition}

\theoremstyle{definition}

\theoremstyle{definition}
\newtheorem{definition}{Definition}
\theoremstyle{definition}
\newtheorem{remark}{Remark}
\theoremstyle{definition}

\global\long\def\op{Permutation Invariance}
\global\long\def\p{Strong Pareto}
\global\long\def\c{Completeness}

\global\long\def\finsum{\succeq_{FSP}}
\global\long\def\sumord{\succeq_{SP}}
\global\long\def\countord{\succeq_{CP}}

\author{Jeremy Goodman\footnote{Johns Hopkins University, Miller Department of Philosophy} \ and Harvey Lederman\footnote{University of Texas at Austin, Philosophy and Population Wellbeing Initiative. Thanks to Geir Asheim, Will Combs, Cian Dorr, Jake Nebel, Marcus Pivato, Dean Spears, and Christian Tarsney for comments and discussion.}}
\date{\today}
\title{Maximal Social Welfare Relations on Infinite Populations Satisfying Permutation Invariance}

\begin{document}

\maketitle

\begin{abstract} \noindent We study social welfare relations (SWRs) on an infinite population. Our main result is a new characterization of a utilitarian SWR as the \emph{largest} SWR (in terms of subset when the weak relation is viewed as a set of pairs) which satisfies Strong Pareto, Permutation Invariance (elsewhere called ``Relative Anonymity'' and ``Isomorphism Invariance''), and a further ``Quasi-Independence'' axiom. 

\


\


\end{abstract}

\section{Introduction}

Social choice over infinite utility streams has a long tradition in economics, starting at least from \citet{ramsey1928mathematical}. The recent literature has taken off from impossibility results concerning tradeoffs between, on the one hand, the (positive) sensitivity of social preferences to individual preferences, and, on the other, the impartiality of social preferences (\citet{basu2003aggregating}). A basic such result displays the inconsistency of Strong Pareto (stating that a distribution which is strictly preferred by some individuals and weakly preferred by all must be socially strictly preferred), and  Anonymity (stating that any distributions related by a permutation of individuals are socially indifferent) (\citet{liedekerke1995impossibility}). More strikingly, Strong Pareto, together with a weak impartiality axiom, Finite Anonymity (that any distributions related by a finitely-supported permutation of individuals are socially indifferent), rules out the existence of any real-valued social welfare function (SWF), given a sufficiently rich set of well-being levels (\citet{basu2003aggregating}). Moreover, while (complete) social welfare orders (SWOs) that satisfy both Strong Pareto and Finite Anonymity do exist (\citet{svensson1980equity}), they must be non-constructive, and so cannot be explicitly described (\citet{fleurbaey2003intertemporal,zame2007can,lauwers2010ordering,dubey2011fleurbaey,dubey2021social}).

In response to these results, much attention has been directed to social welfare relations (SWRs), possibly incomplete pre-orders (transitive and reflexive relations) which satisfy both Strong Pareto and Finite Anonymity. In this vein, a variety of ``utilitarian'' SWRs have been developed, including the ``catching up'' and ``overtaking'' criteria (\citet{atsumi1965neoclassical,weizsacker1965existence,gale1967optimal,brock1970axiomatic,fleurbaey2003intertemporal,asheim2004resolving,basu2007utilitarianism}, for surveys see \citet{pivato2024intergenerational,kamaga2020intergenerational,lauwers2016axiomatic,asheim2010intergenerational}). 

But there are at least two fundamental conceptual questions about these SWRs and their axiomatic basis. 

First, the most commonly studied ``catching up'' and ``overtaking'' SWRs are defined in terms of an intrinsic order on the population, corresponding to the order of ``generations'' in time. But \citet{ramsey1928mathematical} motivates the idea of impartiality by arguing that a person's or generation's position in time should not be relevant for truly impartial social preferences. Perhaps more strikingly, if infinitely many individuals can exist in a single generation, there is no natural sequencing, making it unclear how to apply criteria that depend on such an order. The fact that Finite Anonymity allows for such intuitively order sensitive SWRs suggests that this axiom does not capture the full force of Ramsey's impartiality idea, and motivates the exploration of alternative impartiality axioms. 

Second, existing characterizations of such incomplete SWRs are in terms of their being the \emph{minimal} relations (in terms of set inclusion, when SWRs are viewed as sets of pairs of distributions) satisfying certain axioms. Such characterizations leave open the question whether the SWRs exhibit \emph{too much} incompleteness, and whether ``more decisive'' SWRs may be found that would provide more extensive policy recommendations.

This paper develops an approach to SWRs for infinite populations which answers both of these questions together. We start with an impartiality axiom, which we call Permutation Invariance, but which has also been called Anonymity (\citet[p. 72]{sen1984collective}), Relative Anonymity (\citet{asheim2010generalized}), Isomorphism Invariance (\citet{lauwers2004infinite,jonsson2020consequentialism}), or Qualitativeness (\citet{askell2018diss}), stating that a distribution $w$ is weakly socially preferred to a distribution $v$ if and only if $\pi(w)$ is weakly socially preferred to $\pi(v)$, where $\pi$ is a permutation of the population, which acts in the obvious way on distributions. This axiom is logically independent of Finite Anonymity,\footnote{While Permutation Invariance alone is logically independent of Finite Anonymity, Permutation Invariance together with a Finite Completeness axiom does imply Finite Anonymity. See \citet[n. 4]{asheim2024exploring}, and Remark \ref{finiteanonymity}.} but it intuitively corresponds to a fuller notion of impartiality than Finite Anonymity, since it allows arbitrary permutations, not just those with finite support. As emphasized by \citet{asheim2010generalized} (see also below section \ref{comparisons}), the axiom rules out two of the most widely studied order-sensitive SWRs (``catching up'', ``overtaking''), and thus answers our first question above, providing an  approach to impartiality which is more general than Finite Anonymity (since the former but not the latter rules out intuitively order-sensitive SWRs). (In line with this approach to impartiality, we work throughout with a population that does not have an intrinsic order.) 

Like Finite Anonymity, Permutation Invariance is compatible with Strong Pareto. But unlike Finite Anonymity, Permutation Invariance together with Strong Pareto rules out the existence of any SWOs (Proposition \ref{incompleteness}). This raises the question, which we study here, of whether there  is a \emph{largest} such SWR, that is, an SWR which satisfies the axioms and makes all comparisons made by any relation satisfying the relevant axioms.  As a preliminary result, we show that, in a setting where distributions are assignments of the members of our infinite population to the set $\{0,1\}$, there is such a largest relation among relations which satisfy Strong Pareto and Permutation Invariance (Proposition \ref{twovalued}). 

Our main result (Theorem \ref{discrete}) provides a characterization of a utilitarian SWR in the more general setting where distributions are finite-valued functions from the population to elements of $\mathbb R$. The characterization relies on one further axiom on social preferences, which we call Quasi-Independence, stating that if $w$ is weakly preferred to $v$ then for $\alpha \in [0,1]$ $\alpha w + (1-\alpha) u$ is weakly preferred to $\alpha v + (1-\alpha) u$ (where scalar multiplication and addition are understood pointwise). We show that our utilitarian relation is the \emph{largest} relation satisfying Strong Pareto, Permutation Invariance, and Quasi-Independence.

This result provides an answer to our second question above, of whether standardly-studied incomplete SWRs are sufficiently decisive. By demonstrating that our target SWR (which coincides, in the main setting we study, with that of \citet{basu2007utilitarianism} and the ``utilitarian time-invariant overtaking'' of \citet{asheim2010generalized}) is the largest in this set, we show that, given Permutation Invariance and our other axioms, there cannot be a relation which compares more pairs of distributions. Since these axioms are normatively compelling, this provides significant conceptual support for such SWRs, demonstrating that one cannot hope for greater comparability than they exhibit.

Section \ref{setup} provides setup, shows that Permutation Invariance and Strong Pareto rule out completeness, and defines the notion of a relation being largest. Section \ref{mainresult} presents the main results. Section \ref{extensiontor} provides ancillary results: we show the necessity of Quasi-Independence to the main theorem \ref{discrete}; and that this theorem does not hold where distributions may have infinite range. Section \ref{relatedwork} discusses related work. Section \ref{comparisons} recalls how Permutation Invariance rules out the standard ``overtaking'' and ``catching up'' SWRs; section \ref{characterizations} compares our maximality-based characterization to minimality-based ones; and section \ref{errors} discusses a result of \citet{lauwers2004infinite} that appears to cover related ground to ours, characterizing the order we study here using a weak impartiality assumption. We show that the claimed result in that paper is incorrect, but that it can be repaired, and consider the relationship between the corrected result and our main theorem.

\section{Setup and Motivation}\label{setup}

\subsection{Notation and Basic Definitions}

Throughout $X$ is a fixed countably infinite set, thought of as the set of individuals. As emphasized earlier, no order is assumed on this set. The set of \emph{worlds} is $W \subseteq \mathbb R^X$, where the real numbers are understood as the set of welfare levels for individuals.

A function is \emph{finite-valued} if and only the cardinality of its range is finite. We write $W_F$ for the set of such finite-valued functions in $\mathbb R^X$. Our main result concerns the setting where $W=W_F$. For any finite subset of individuals and any assignment of those individuals to real welfare levels, there are infinitely many elements of $W_F$ which realize that distribution for that set of individuals. Note that while the range of each function in this set has finite cardinality, the union of the ranges of the functions in the set is $\mathbb R$.

Throughout we study properties of SWRs i.e. reflexive, transitive, relations on $W$. As usual we write $w \succ v$ when $w \succeq v$ and $v \not \succeq w$ and $w \sim v$ when $w \succeq v$ and $v\succeq w$. Two outcomes $w, v$ are \emph{incomparable} in an order $\succeq$, denoted $w \perp v$ if and only if both $w \not \succeq v$ and $v \not \succeq w$. 
\subsection{Strong Pareto and Permutation Invariance}

We study the consequences of a strong (positive) sensitivity or efficiency axiom:

\begin{description}
\item[Axiom 1. Strong Pareto] For all $w, v \in W$, if for all $x \in X$, $w(x) \geq v(x)$ and for some $x \in X, w(x)>v(x)$, then $w\succ v$.
\end{description}

We combine this axiom, as promised, with an impartiality axiom. Given a permutation $\pi$ of $X$ we take $\pi(w)$ to be defined by $\pi(w)(x)=w(\pi(x))$ for all $x \in X$.

The following axiom has been widely discussed:
\begin{description}
\item[Anonymity] $w \sim \pi (w)$. 
\end{description}
We will not assume this axiom, because it is inconsistent with Strong Pareto, provided $|L|\geq 2$ (\citet{liedekerke1995impossibility}; for strengthenings of his result, see \citet[proposition 1]{liedekerke1997sacrificing}, \citet{lauwers1998intertemporal}, \citet{basu2003aggregating}, \citet[Theorem 1]{fleurbaey2003intertemporal}, \citet{mitra2007existence}). To see this, let $L=\{0,1\}$ and understand populations as characteristic functions of subsets of $X$, namely, the subset of people with the greater well-being level. The inconsistency then follows from the fact that any infinite, coinfinite subset of $X$, $A$ can be injected into a strict subset of itself, $B$ by a permutation of $X$.

In place of Anonymity, we impose the weaker:
\begin{description}
\item[Axiom 2. Permutation Invariance] For any, $w, v \in W$, $w \succeq v$ if and only if $\pi (w) \succeq \pi (v)$.
\end{description}
This axiom was called ``Anonymity'' in \citet{sen1984collective} and ``Relative Anonymity'' in \citet{asheim2010generalized}. We call it ``Permutation Invariance'' to highlight the property of the order with respect to permutations. Like Anonymity, Permutation Invariance is an axiom of ``impartiality''. It is implied by Anonymity, but is strictly weaker than it.

Conceptual motivation for Permutation Invariance as opposed to Anonymity is provided by natural, intuitively impartial binary relations on distributions, which nevertheless violate Anonymity. Consider for instance the ``Pareto preorder'', which ranks one distribution as better than another if and only if the first Pareto dominates the second. There is an intuitive sense in which this relation is ``impartial'': it does not give special consideration to individuals (on their own, or collectively). But this relation violates Anonymity in the infinite setting, as shown by the argument above: $w$ may (strictly) Pareto-dominate $\pi(w)$ (because an infinite set may be injected into a strict subset of itself by a permutation). The Pareto preorder does however satisfy Permutation Invariance. This fact motivates the use of Permutation Invariance as an impartiality axiom.

\subsection{Incompleteness and Maximality}

While Strong Pareto and Permutation Invariance are consistent, they are not consistent with the following standard axiom:
\begin{description}
    \item[\c{}:] For all worlds $w$ and $v$, either $w\succeq v$ or $v\succeq w$.
\end{description}

\begin{proposition}[\citet{askell2018diss}]\label{incompleteness} Let $W= \{0,1\}^X$. \p{} and \op{} are inconsistent with \c{}.
\end{proposition}

\noindent While \citet{askell2018diss} gives an argument for this result, it has not been discussed in the economics literature.

\begin{proof}
Let $A$ and $B$ be disjoint infinite subsets of $X$, and $A^-$ an infinite proper subset of $A$. Let $w, w^-, v$ be distributions such that: 
\begin{itemize}
\item for all $x \in A$ $w(x)=1$, and otherwise $w(x)=0$;
\item for all $x \in A^-$, $w^-(x) =1$, and otherwise $w^-(x)=0$
\item for all $x \in B$ $v(x)=1$ and otherwise $v(x)=0$.
\end{itemize}
By \p{} $w\succ w^-$. 

By \c{}, $w\succ v$, or $v \succ w$, or $w \sim v$. There is a permutation $\pi$ so that $x \in A$ iff $\pi(x)\in B$, and $x \in B$ iff $\pi(x) \in A$, so that $\pi(w)=v$ and $\pi(v)=w$. If $w \succ v$, then by \op{}, $\pi(w) \succ \pi(v)$, i.e. $v \succ w$, a contradiction. Similarly if $v \succ w$, then $\pi(v) \succ \pi(w)$, i.e. $w \succ v$, a contradiction.

If $w \sim v$, then since $w\succ w^-$, $v \succ w^-$. There is a permutation $\pi'$ so that $x \in A^-$ iff $\pi'(x)\in B$, and $x \in B$ iff $\pi'(x) \in A^-$ so that $\pi'(v)=w^-$ and $\pi'(w^-)=v$. Given $v \succ w^-$, by  \op{}, $\pi'(v)\succ \pi'(w^-)$, i.e. $w^-\succ v$, a contradiction.
\end{proof}

One reaction to this result would be to hold that either Strong Pareto or Permutation Invariance must be rejected so that Completeness can be upheld. But in our view, given the wide array of known impossibility results for infinite populations, it is not obvious how to proceed, and it is at least open that Completeness does fail in this setting. Thus, in this paper, we assume both Strong Pareto and Permutation Invariance, and study SWRs as opposed to SWOs.

A concern about SWRs is that they may have ``too much'' incompleteness. For example, let $W=\{0,1\}^X$ and consider again the Pareto preorder that is, the preorder according to which $w \succeq v$ iff $\{ x | w(x)=1\} \supseteq \{ x | v(x) =1\}$. This preorder satisfies both Strong Pareto and Permutation Invariance. But the preorder is very weak: it fails to rank a world in which any finite number $n$ individuals receive wellbeing level $1$ (and all else $0$) over a world in which one (distinct) person receives wellbeing level $1$ (and all else $0$).

This example shows that our SWRs can have a great deal of incompleteness, intuitively, much ``more'' than is required. This motivates our question here, about \emph{maximal} orders consistent with our axioms. Formally:

\begin{definition}[Extension and Maximality] A relation $\succeq$ \emph{weakly extends} a relation $\succeq'$ if and only if, whenever $x \succeq' y$, $x\succeq y$. A relation $\succeq$ is \emph{maximal} within a set of relations $\mathcal R$ if and only if $\succeq \in \mathcal R$ and there is no relation in $\mathcal R$ that weakly extends $\succeq$.
\end{definition}

Maximal relations in this sense may fail to be unique. Since our interest is in using a property related to maximality to provide a characterization result, our main results will focus on the more demanding property, of being largest, which does implies uniqueness: 

\begin{definition}[Largest]
A binary relation $\succeq$ is the \emph{largest} relation in a set of binary relations $\mathcal R$ if and only if $\succeq \in \mathcal R$ and $\succeq$ weakly extends every $\succeq' \in \mathcal R$. 
\end{definition}
If a largest relation in a set of binary relations exists, it is guaranteed to be the \emph{unique} maximal relation within the set. Our characterization results about largest SWRs will thus demonstrate that there is a unique such maximal relation.

Our notion of ``largest'' uses \emph{weak} extensions rather than extensions, where a preorder $\succeq'$ is said to \emph{extend} a preorder $\succeq$ if and only if, whenever $x \succeq y$, $x\succeq'y$, \emph{and} whenever $x \succ y$, $x \succ'y$. If a largest relation within a set of relations exists, no relation in the set extends it. But it is not guaranteed to be the unique relation which cannot be extended.

We study weak extension rather than extension here because our axioms will be compatible with SWRs that make ``deviant'' strict comparisons where a weak comparison is more intuitive. It is compatible with Strong Pareto and Permutation Invariance that for $a\neq b \in X$, if $w(a)=1$ and for all other $x \in X$ $w(x)=0$, while $v(b)=1$ and for all other $x\in X$, $v(x)=0$, $w\succ v$. But it is also compatible that $v\succ w$. This simple example already establishes there is not a unique relation that cannot be (strongly) extended. Moreover, these strict comparisons are implausible. Considering the largest relation in our sense, as a weak extension of all other relations both allows for an existence result, and delivers a more intuitive SWR, since deviant strict comparisons can be ``smoothed out'' into indifference by the largest relation. (We discuss this point again in connection with remark \ref{vallentynerem}.)

\section{Main Result}\label{mainresult}

Our main result establishes the existence of a largest relation when $W=W_F$. We first define the preorder we will show to be maximal (\ref{fsoso}), then state one further axiom (\ref{axioms}), then state the theorem (\ref{proof}).

\subsection{The Sum Preorder and the Finite Sum Preorder}\label{fsoso}

We will show the maximality of an SWR studied axiomatically by \citet{lauwers2004infinite} (cf. \citet{vallentyne1997infinite}). Recall that the sum of an infinite sequence converges (diverges) \emph{unconditionally} if and only if it converges (diverges) regardless of the order in which terms are summed. Such a sum converges conditionally if and only if it converges on some but not all orderings of its terms.
\begin{definition}[Sum Preorder ($\sumord$)] For all $w, v \in W$, $w \sumord v$ if and only if $\sum_{x \in X} w(x)-v(x)$ converges unconditionally to $r \geq 0$, or diverges unconditionally to positive infinity.
\end{definition}

\begin{remark}\label{finsumremark} We focus throughout explicitly on $\sumord$, but our main result holds equivalently for an alternative preorder, which \citet{asheim2010generalized} call \emph{utilitarian time-invariant overtaking}:
\begin{description}
\item[Finite Sum Preorder ($\finsum$)] For all $w, v \in W$, $w \finsum v$ if and only if there is a finite set $A \subset X$ such that for every finite $B\supseteq A$, $\sum_{x \in B} w(x)-v(x)\geq 0$. \end{description}
When $W=\mathbb R^X$, and preorders are viewed as sets of pairs $\finsum \subsetneq \sumord$ (see below in remark \ref{vallentynerem}). But when $W=W_F$, the two orders in fact coincide. Since our result will be proven in this setting, it applies directly to both. 

Similarly, when $W=W_F$, the ``utilitarian'' criterion of \citet{basu2007utilitarianism} coincides with $\finsum$, though the latter is again stronger in a general setting where $W=\mathbb R^X$ (see \citet[\S1]{asheim2010generalized} for an example).
\end{remark}

\subsection{Quasi-Independence}\label{axioms}

Our main result uses a further axiom in addition to \p{} and \op{}. To state the axiom, we will use scalar multiplication and addition on worlds with these operations understood pointwise. So, for instance, for worlds $w,v$, $w+v$ is defined so that for all $x$, $(w+v)(x)=w(x)+v(x)$, and for real $\alpha$, $\alpha w$ is defined so that for all $x$ $(\alpha w)(x)=\alpha (w(x))$. The axiom is then as follows:
\begin{description}
\item[Axiom 3. Quasi-Independence] For any $w, v, u \in W$, if $w \succeq v$, then for any $\alpha\in [0,1]$ $\alpha w + (1-\alpha) u \succeq \alpha v + (1-\alpha)u$.
\end{description}
It will be important later that if $w, u\in W_F$, then for any $\alpha \in [0,1]$ $\alpha w_1 + (1-\alpha) w_2 \in W_F$ as well.

Quasi-Independence looks similar to the right-to-left direction of the standard Independence axiom. In this sense it is significantly weaker than Independence itself. However, since we are working not with lotteries but directly with worlds, and since we assume that multiplication and addition of worlds are understood pointwise, the assumption is structurally different than the Independence axiom. We have adopted the name to highlight the relationship between the two axioms, while indicating the important differences between them.

In this paper, we study preorders on worlds. A common setting considers preorders defined instead over lotteries on worlds. In that more general setting, Axiom 3 (Quasi-Independence) can be motivated by the fact that it is a consequence of the following two axioms (where we write $\Delta(W)$ for the set of finite-support lotteries over $W$ and assume that preferences are defined on the set of such lotteries, with worlds understood as degenerate lotteries):
\begin{description}
\item[Weak Independence] For any $F, G \in \Delta (W)$ if $w \succeq v$, $F(w)=G(v)=\alpha$, and $F(u)=G(u)=(1-\alpha)$ then $F\succeq G$.
\item[Ex Ante Indifference] For finite-support $F, G \in \Delta (W)$, if for all $x\in X$, $EV(F(x))=EV(G(x))$, then $F \sim G$.
\end{description}
Here, for a lottery $L$ and individual $x \in X$, $EV(L(x))= \sum_{w \in W} L(w)w(x)$.

The first of these, as indicated earlier, is a significant weakening of the Independence Axiom of \citet{von1944theory}; it is just the ``if'' direction of that axiom. The second axiom, a version of which is used pivotally in \citet{harsanyi1955cardinal}, is a strong, and less familiar assumption. A famous example due to \citet{diamond1967cardinal} illustrates directly how the axiom rules out ex ante egalitarianism. A further example of \citet{myerson1981utilitarianism} shows how it rules out ex post egalitarianism. Our Quasi-Independence assumption similarly rules out these forms of egalitarianism. But we make it as a natural starting point, to explore its consequences in the infinite setting. As we discuss later, we think it is worth exploring whether further results about largest relations can be given using weaker assumptions than this one, or alternatives to it.

\subsection{Main Results}\label{proof}

In the simple setting with only two wellbeing levels, Axiom 1 (Strong Pareto) and 2 (Permutation Invariance) suffice for a largest preorder:

\begin{proposition}\label{twovalued} Let $W= \{0,1\}^X$. $\sumord$ is the largest relation in the set of preorders on $W$ which satisfy Strong Pareto and Permutation Invariance.\end{proposition}

In the more general setting of finite-valued worlds, Axioms 1-3 suffice:

\begin{theorem}\label{discrete} Let $W=W_F$. $\sumord$ is the largest relation in the set of preorders on $W$ which satisfy Strong Pareto, Permutation Invariance, and Quasi-Independence. 
\end{theorem}

The proofs of both of these claims are given in appendices. As mentioned earlier, an immediate corollary of each of these results is that the relevant preorders are the unique maximal relations in the relevant sets.

\begin{remark}[Finite Anonymity]\label{finiteanonymity} Finite Anonymity is a consequence of each of these results, but not an assumption of them. The reason is roughly as follows. It is compatible with our axioms that any worlds which differ on only finitely many individuals be ranked. So if a largest relation exists it must compare all such pairs of worlds. Given such ``Finite Completeness'', Permutation Invariance implies Finite Anonymity (\citet[n. 4]{asheim2024exploring}).\end{remark}

\section{Limitations of the Main Result}\label{extensiontor}

In this section, we illustrate limitations of theorem \ref{discrete}. First, we show that Axiom 3 (Quasi-Independence) is necessary for the main result (section \ref{cd}). Second, we show that the Sum Preorder is not maximal if we allow infinite-valued worlds, i.e. if we let $W=\mathbb R^X$ (proposition \ref{nonmaximality}). 

\subsection{The Necessity of Quasi-Independence}\label{cd}

Given $W=\mathbb R^X$,  $\sumord{}$ is not maximal, given only Strong Pareto, and Permutation Invariance.

\begin{definition}[Sum Preorder plus differences ($\trianglerighteq$)]
For any $w, v$, $w^+_v=\{ x| w(x)>v(x)\}$. Let $w \trianglerighteq v$ if and only if either $w \sumord{} v$, or if, for any finite subsets $Y, Z$ with $|Y|=|Z|$, $Y \subseteq w^+_v$ and $Z \subseteq v^+_w$, $\sum_{y \in Y} w(y)-v(y) \geq \sum_{z \in Z} v(z) - w(z)$.
\end{definition}

\begin{proposition} Let $W=\{0,1,2\}^X$. $\trianglerighteq$ is a preorder which satisfies Strong Pareto and Permutation Invariance, and
which properly weakly extends $\sumord$. \end{proposition}

\begin{proof} Given an infinite, coinfinite $A\subset X$, let $w(x)=2$ for all $x \in A$ and otherwise $0$, while $v(x)=1$ for all $x \in X \setminus A$ and otherwise $0$, we have $w \trianglerighteq v$, but $w \not \sumord{} v$. This relation is transitive and reflexive, and satisfies \p{} and \op{} (because the sums are defined order-independently). (The proof extends to $W=L^\mathbb R$ for any $L\subseteq R$ such that $|L| \geq 3$.) \end{proof}

\subsection{Non-maximality of the Sum Preorder Beyond $W_F$}\label{nonmaximalsection}

Next we show that the maximality result for $W_F$ cannot be extended to the setting where $W=\mathbb R^X$, i.e. where infinite-valued worlds are allowed. First we note that in this setting $\finsum$ and $\sumord$ no longer coincide:

\begin{remark}[\citet{vallentyne1997infinite}]\label{vallentynerem} Let $W=\mathbb R^X$. $\sumord$ weakly extends $\finsum$ (see above remark \ref{finsumremark} for the definition). Fixing an enumeration of elements of $X$ we write distributions as sequences, with entries corresponding to the well-being level of the corresponding position in the enumeration. Now consider $w=\langle1, 0, 0, 0...\rangle$ and $v=\langle 1/2, 1/4, 1/8,...\rangle$. In this case $v \not \finsum w$, however $v \sumord w$.\end{remark}

\noindent This remark further motivates the use of weak extensions (as opposed to extensions) in the notion of a relation being largest. Intuitively, the failure of $\finsum$ to rank $w$ and $v$ as equivalent in this example shows that it is too weak. But an extension of $\finsum$ would be required to preserve this weakness. $\sumord$ can deliver the intuitive verdict precisely because it is only a weak extension of $\finsum$, not an extension.

\noindent We now show that $ \sumord$ is neither maximal nor largest:

\begin{definition}[Sum of Differences and Convergent Divergences ($\blacktriangleq$)]
For all $w, v \in W$, $w \blacktriangleq v$ if and only if
\begin{itemize}
\item[(i)] $w \sumord{} v$ or
\item[(ii)] \begin{itemize}
\item[(a)] $w \not \sumord v$,
\item[(b)] for some $c > 0$ and infinite $A \subseteq X$, for all $x \in A$ $w(x)-v(x)>c$, and
\item[(c)] there is no $d>0$ and infinite $B \subseteq X$ so that for all $x \in B$, $v(x)-w(x)>d$.
\end{itemize}
\end{itemize} 
\end{definition}

\begin{proposition}\label{nonmaximality} Let $W=\mathbb R^X$. 
\begin{enumerate}
\item $\blacktriangleq$ is a preorder satisfying Strong Pareto, Permutation Invariance, and Quasi-Independence.
\item $\blacktriangleq$ strictly weakly extends  $\sumord{}$.
\end{enumerate}
 \end{proposition}

\begin{proof} 
1 is tedious; the proof is relegated to an appendix. For 2, let $A$, $B$ be infinite disjoint subsets of $X$ and fix an enumeration of $B$, $b_1, b_2...$. Let $w(x)=1$ if $x \in A$, and $0$ otherwise. Let $v(b_n)=1/n$, and $v(x)=0$ otherwise. Then $w \blacktriangleq v$, but not $w \sumord v$. 
\end{proof}

\begin{remark} Nothing in this argument depends essentially on the use of $\mathbb R$. For instance, we can use the same technique to show that the relation is not maximal on $\mathbb Q^X$. \end{remark}

\section{Related Work}\label{relatedwork}

\subsection{Inconsistency of Catching Up and Overtaking with Permutation Invariance}\label{comparisons}

Throughout this subsection, fix an enumeration of $X$. The usual definitions of Catching Up and Overtaking are as follows (\citet{atsumi1965neoclassical,weizsacker1965existence,brock1970axiomatic,fleurbaey2003intertemporal,asheim2004resolving}):\footnote{A different use of this terminology can be found in \citet{gale1967optimal}, \citet[\S3]{jonsson2018limit}. We follow the definitions of \citet{pivato2024intergenerational} and \citet{basu2007utilitarianism}.}
\begin{description}
\item[Catching Up] For all $w, v \in W$, $w \succeq_{CU} v$ iff for some $\overline N$, for all $N \geq \overline N$ $$\sum_{x_i: 1 \leq i \leq N} w(x_i) \geq\sum_{x_i: 1 \leq i \leq N} v(x_i)$$
\item[Overtaking] For all $w, v \in W$, \begin{itemize}
\item $w \succ_{O} v$  iff for some $\overline N$ for all $N \geq \overline N$ $\sum_{x_i: 1 \leq i \leq N} w(x_i) > \sum_{x_i: 1 \leq i \leq N} v(x_n);$ 
\item $w \sim_{O} v$  iff for some $\overline N$, for all $N \geq \overline N$ $\sum_{x_i: 1 \leq i \leq N} w(x_i) = \sum_{x_i: 1 \leq i \leq N} v(x_i).$
\end{itemize}
\end{description}

\noindent Both of these preorders were defined using the natural temporal order of generations. It is not surprising, then, that both of them violate Permutation Invariance (see \citet{asheim2010generalized}). In fact they do so, even if $W= \{0,1\}^X$. To see this, let $w=\langle 1, 0, 1, 0, 1, 0 \dots \rangle$, and $v= \langle 0, 0, 0, 1, 0, 1, 0, 1, 0, 1 \dots \rangle$. Then $w \succ_O v$ and $w \succ_{CU} v$. But Permutation Invariance rules out this strict comparison, because there is a permutation $\pi$ such that $\pi(w)=v$ and $\pi(v)=w$ (so that $w\succ v$, would imply $v \succ w$, a contradiction). As discussed earlier, this distinguishes our setting from the setting of these two commonly studies SWRs.

\subsection{Comparison to Other Characterization Results}\label{characterizations}

\citet{basu2007utilitarianism} introduce a ``utilitarian'' social welfare relation, which is weaker than $\finsum$ if $W=\mathbb R^X$ (see \citet[\S 1]{asheim2010generalized}) but coincides in the case where $W=W_F$. \citet[Theorem 1]{basu2007utilitarianism} shows that their preorder is the minimal relation which satisfies Strong Pareto, Finite Anonymity and:
\begin{description}
\item[Partial Unit Comparability] For all $w,v\in W$ if $\{x | w(x)\neq v(x)\}$ is finite, $w\succeq v$, and $u\in \mathbb R^X$ is such that $w+u, y + u\in W$ and $u(x) = 0$ if $w(x)=v(x)$, then $x + u\succeq y + u$. 
\end{description}
 \citet[Theorem 6]{jonsson2015utilitarianism} provide an analogous characterization, using a weaker version of Partial Unit Comparability. Our characterization is on a less rich domain, but we do not use any analogue of Partial Unit Comparability, and our Permutation Invariance is neither necessary nor sufficient for Finite Anonymity. Moreover, our characterization is of the relation as maximal. 

\citet{basu2007utilitarianism} argue that their SWR gives correct results, in spite of yielding greater incomparability by contrast to the more standard Overtaking and Catching Up SWRs. Our result can be seen as supporting their arguments, in an axiomatic framework, by appeal to Permutation Invariance, since we show that no further comparability on $W_F$ is compatible with our axioms.

\citet[Corollary 1]{asheim2010generalized} characterize $\finsum$ with $L=\mathbb R$. They show that it is the \emph{minimal} relation satisfying Finite Anonymity, Finite Pareto (a weakening of Strong Pareto), Partial Unit Comparability (which they call ``Finite Translation Scale Invariance''):
\begin{description}
\item[Time Invariant Preference Continuity (IPC)] For all $w,v\in W$, if there exists is a finite $Z \subset X$ such that, for all finite $Y\subseteq Z$, $(w_Y, v_{X-Y})\succeq y$ then $x\succeq y$ (where $(w_Y, v_{X-Y})$ is defined so that for all $x\in Y$ $(w_Y, v_{X-Y})(x)=w(x)$, and for all other $x$, $(w_Y, v_{X-Y})(x)=v(x)$).
\end{description}
Our characterization is again on a less rich domain, but does not use either Partial Unit Comparability, or IPC, the latter of which especially is a strong axiom. Our result also uses Permutation Invariance, which is incomparable in strength with Finite Anonymity. Finally, we characterize this relation as maximal, not as minimal. As just noted in connection to the arguments of \citet{basu2007utilitarianism}, we think this is of special conceptual interest.

\citet[Theorem 3]{fleurbaey2003intertemporal} and  \citet[Proposition 3]{asheim2004resolving} characterize Catching Up and Overtaking, but as we have noted these SWRs are incompatible with Permutation Invariance.

\subsection{Comparison to Lauwers and Vallentyne}\label{errors}

\citet{lauwers2004infinite} appears to provide a characterization of our preorder in a more general setting, which we now discuss.

\citet[Theorem 5]{lauwers2004infinite} use a principle they call Restricted Transfer. The main idea of the principle is introduced in the following quotation:
\begin{quote} The condition appeals to the notion of a restricted transfer, which is (1) a transfer of a positive amount of value from a location with positive value to a location with negative value such that (2) after the transfer, the donor location still has non-negative value and the recipient location still has non-positive value. For example, the move from $\langle -1, 3\rangle$ to $\langle 0,2\rangle$ is a restricted transfer, but the move from $\langle -1, 3\rangle$ to $\langle 1,1\rangle$ is not. (p. 323)
\end{quote}

Their principle is then as follows.

\begin{description} 
\item[Restricted Transfers:] If locations have no natural structure, then, for any three worlds, $U$, $U^*$, and $V$, if (1) $U$ is better than $V$, and (2) $U^*$ is obtainable from $U$ by some (possibly infinite) number of restricted transfers, then $U^*$ is better than $V$. (p. 323)\end{description}

They claim that the Sum Preorder satisfies Restricted Transfer. But Restricted Transfer is inconsistent with the Sum Preorder, as the following fact establishes.\footnote{\citet[n. 188]{askell2018diss} also observes that Lauwers and Vallentyne's stated principle, which allows transfers from one person with positive welfare, to multiple people with negative welfare is inconsistent with Strong Pareto. But she does not remark that even her restricted principle is inconsistent. The restricted principle is employed in her Result 16, which, while correct, has assumptions that cannot be jointly satisfied.}

\begin{fact}\label{lv} Let $W\supseteq \{-2,0,1, 2\}^X$. The Sum Preorder is inconsistent with Restricted Transfer.\end{fact}

\begin{proof} Let $A$ be an infinite coinfinite subset of $X$. Let $a,b,c\in W$ be defined so that:
\begin{itemize}
\item $a(x)=2$ if $x \in A$ and $a(x)=-2$ otherwise.
\item $b(x)=1$ if $x \in A$ and $b(x)=-2$ otherwise
\item $c(x)=0$ for all $x \in X$.
\end{itemize}

$a \succ_{SO} b$. By Restricted Transfer it would follow that $c \succ_{SO} b$.  But $b \perp_{SO} c$.
 \end{proof}

Lauwers and Vallentyne's proof of their Theorem 5 establishes a different fact than the consistency of the Sum Preorder with Restricted Transfers, and suggests that the proof was intended to demonstrate a somewhat different claim. The condition they seem to have in mind is as follows (where $\mathbf 0$ is the constant function from $X$ to $0$):

\begin{description}
\item[Restricted Transfers (Corrected):] If locations have no natural structure, then, for any three worlds, $U$, and $U^*$, if (1) $U$ is better than $\mathbf 0$, and (2) $U^*$ is obtainable from $U$ by some (possibly infinite) number of restricted transfers, then $U^*$ is better than $\mathbf 0$.\footnote{Thanks to an anonymous referee for suggesting this criterion, which allowed us to greatly simplify our discussion here.}
\end{description}

Given this much weaker principle, which applies only to comparisons to ``the zero world'', the proof of their result goes through. The corrected result provides a characterization of the Sum Preorder, using the following additional axioms:
\begin{description}
\item[Weak Pareto] For all $w, v \in W$ if for all $x \in X$ $w(x) \geq v(x)$ then $w \succeq v$.
\item[Zero Independence] For all $w, v \in W$, $w \succeq v$ if and only if $w-v \succeq \mathbf 0$.
\item[Sum] For all $w, v \in W$ if $\sum_{ x\in X} w(x)$ and $\sum_{x \in X} v(x)$ are unconditionally convergent, then $w \succeq v$ iff $\sum_{ x\in X} w(x) \geq \sum_{x \in X} v(x)$.
\end{description}

Zero Independence is a significant strengthening of Partial Unit Comparability, above.

\begin{proposition}[\citet{lauwers2004infinite}, Theorem 5 corrected] The Sum Preorder is the unique preorder satisfying Restricted Transfers (Corrected), Weak Pareto, Zero Independence, and Sum.\end{proposition}

The proof they give in their paper is a proof of this claim.

Unlike characterizations more standard in the economics literature, this result does not appeal to minimality of the relevant SWR. In this sense it is a precedent for our work. It goes further than our work in giving a characterization of the order on the whole of $\mathbb R^X$. However, we do not see how to motivate Restricted Transfer (Corrected) in a way that is independent of the motivations for the Sum Preorder itself. Indeed, it seems to us that Restricted Transfer (Corrected) simply restates the idea that $w \succeq v$ only if the sum of the positive values of $w-v$ ``outweighs'' the sum of the positive values of $v-w$. By contrast, Quasi-Independence can be independently motivated. (Note that we also do not use a version of Zero Independence.)

\section{Conclusion}\label{conclusion}

This paper has studied the implications of the conjunction of Strong Pareto and Permutation Invariance in infinite populations, where they imply the incompleteness of social preferences. We have asked whether, and under what conditions, there are largest SWRs satisfying these two axioms. With two welfare levels, the Sum Preorder is the largest SWR satisfying Strong Pareto and Permutation Invariance (Proposition \ref{twovalued}). With finite-valued worlds, the Sum Preorder is the largest SWR satisfying Strong Pareto, Permutation Invariance, and a further Quasi-Independence axiom (Theorem \ref{discrete}).

We see at least three directions for future work.

First, most obviously, we have not resolved the question of whether, if $W=\mathbb R^X$ there exists a largest preorder which satisfies our axioms. We have provided some partial results on this question, and shown that the Sum Preorder is not largest or even maximal on this set. We hope our results here will spur further work on the topic.

Second, our main result effectively assumes a utilitarian criterion for ranking worlds which differ on finite sets of individuals. This is required by Strong Pareto, Permutation Invariance and Quasi-Independence (as shown by Lemma \ref{lemma1}, proved below in the appendix). It is an interesting question whether various egalitarian conditions on ranking such finite-difference worlds (e.g. those studied in \citet{asheim2010generalized}) could also be the subject of analogous maximality results, which would require alternative axioms in place of Quasi-Independence.

Third, there are a wide array of questions about the ``size'' of the relation we have established as largest on $W_F$. For instance: what, in a natural measure on the space of pairs of worlds, is the measure of comparable pairs? Is this set dense topologically? We have not begun to address these questions here.

\begin{appendices}

\section{Proofs}

 \subsection{Proof of Proposition \ref{twovalued}}
 
 \begin{proposition*} Let $W= \{0,1\}^X$. $\sumord$ is the largest preorder satisfying both \p{} and \op{}.\end{proposition*}
 
\noindent We prove the proposition by considering a simpler preorder, which is equivalent on this restricted domain of worlds:
 
 \begin{definition}[Counting Preorder ($\countord$)] Let $L=\{0,1\}$, $w, v \in W$, and $1_w=\{x | w(x)=1\}$, $1_v=\{x | v(x)=1\}$. $w  \countord{} v$ if and only if $|1_w-1_v| \geq |1_v-1_w|$ and the latter is finite.\footnote{This can be extended to higher cardinalities of infinity by adding a second disjunct: that $w$ is at least as good as $v$ iff either the above condition is satisfied or $|1_w-1_v|>|1_w-1_v|$.}\end{definition}
 
 \begin{remark} If $W=\{0,1\}^X$, $\countord=\finsum= \sumord$.\end{remark}

\begin{proof} Any transitive reflexive relation $\sqsupseteq$ (with strict relation $\sqsupset$)  which holds between sets not included in $\countord$ will violate either \p{} or \op{}. This shows that every reflexive and transitive relation satisfying \p{} and \op{} is a subset of our relation and hence that $\countord$ is largest. 

The proof of proposition \ref{incompleteness} already shows that if $|A-B|$ and $|B-A|$ are both infinite, then they cannot be comparable in such a relation consistently with our assumptions. So it only remains to show that we cannot have $A\sqsupseteq B$ while $|A-B|<|B-A|$, with $|A-B|$ finite.  

Suppose for reductio that $A\sqsupseteq B$ while $|A-B|<|B-A|$ and $|A-B|$ is finite. Since $A-B$ is disjoint from $B-A$ by definition, and given that $|A-B|$ is finite and strictly less than $|B-A|$, there is a permutation $f$ which moves every element of $A-B$ to an element of $B-A$ and moves the relevant element of $B-A$ to its preimage in $A-B$, while mapping every other $x \in X$ to itself. In particular, writing $A \cap B$ as $O$ and allowing permutations to act on sets in the obvious way, we have that $f(O)=O$ (because for all $x \in O$, $f(x)=x$). By \op{}, given that $A \sqsupseteq B$, $f(A) \sqsupseteq f(B)$, or equivalently $f(A-B) \cup O \sqsupseteq f(B-A) \cup O$. But given the choice of $f$, $f(B-A)$ is a strict superset of $A-B$, and so $f(B-A) \cup O$ (that is, $f(B)$), is a strict superset of $(A-B) \cup O$, that is, $A$. As a consequence, \p{} implies that $f(B) \sqsupset A$. Similarly, $f(A-B)$ is a strict subset of $B-A$, and so $(f(A-B) \cup O)$ (that is, $f(A)$) is a strict subset of $((B-A) \cup O)$, that is, $B$. As a consequence, \p{} implies that $B \sqsupset f(A)$. Putting this together, we have $f(A) \sqsupseteq f(B) \sqsupset A \sqsupseteq B \sqsupset f(A)$, which implies (given transitivity and irreflexivity of $\sqsupset$, which follows by definition given the transitivity and reflexivity $\sqsupseteq$) that $f(A) \sqsupset f(A)$, contradicting the reflexivity of $\sqsupseteq$. \end{proof}

\begin{remark} An easy generalization of this argument shows that if $\succeq$ satisfies \p{} and \op{} then for any infinite, disjoint $A$, $B$ and real $k,l$ with $k\neq l$, if for all $x \in A$, $w(x)=k$, and is otherwise $l$, while for all $x \in B$ $v(x)=k$ and is otherwise $l$, $w \perp v$.\end{remark} 

\subsection{Proof of Theorem \ref{discrete}}

The theorem is proved in three steps.

\subsubsection{First Step: Partial Utilitarianism From Quasi-Independence}

The first main lemma shows that Quasi-Independence (against the background of our other assumptions) implies that any comparability in a preorder respects the sums of finite differences.

Before stating that Lemma, we first prove that Quasi-Independence implies a slightly easier to use formulation:

\begin{lemma}\label{convdom} Let $W = W_F$. If $\succeq$ is a preorder satisfying Axiom 3 (Quasi-Independence) then it satisfies:
\begin{description}
\item[Convex Dominance] For any $w_1, \dots, w_n, v_1, \dots, v_n \in W$, if for all $i$ with $1 \leq i \leq n$ $w_i \succeq v_i$, then for any $\alpha_1,\dots,\alpha_n$ such that $\sum_{1\leq i\leq n} \alpha_i =1$, $\sum_{1\leq i\leq n} \alpha_i w_i \succeq \sum_{1\leq i\leq n} \alpha_i v_i$. 
\end{description}
\end{lemma}

\begin{proof} The proof is by an easy induction. 
\end{proof}

In what follows we sometimes directly invoke Convex Dominance.

The main lemma of this stage is as follows:

\begin{lemma}\label{lemma1} Let $W=W_F$. If a preorder $\succeq$ satisfies Axioms 1-3, then if $\{ x | w(x) \neq v(x)\}$ and $w \succeq v$, then  $\sum_{x\in X} w(x)-v(x) \geq 0$. 
\end{lemma}

This result can be generalized to any subset of $\mathbb R^X$ (including the whole set), provided the set is closed under convex combinations and finite permutations.

\begin{proof} Let $w, v$ be worlds such that $Z= \{ x | w(x) \neq v(x) \}$ is finite. Let $\Pi_Z=\{ \pi | \pi$ permutes $Z$ and fixes $X-Z\}$. $\Pi_Z$ is finite; let $n$ be the number of its members. Fix an enumeration of its elements as $\pi_1 \dots \pi_n$. Let $w_1=\pi_1(w)$, $w_2= \pi_2(w)$ and so on, and similarly for $v$. Let $L_w$ be $\sum_{1 \leq i \leq n} 1/n w_i$ and $L_v$ be $\sum_{1 \leq i \leq n} 1/n w_i$. Note that, if $x \in Z$, $L_w(x)=\frac{ \sum_{x \in Z} w(x)}{|Z|}$ and otherwise $L_w(x)=w(x)$. Similarly if $x \in Z$, $L_v(x)=\frac{ \sum_{x \in Z} v(x)}{|Z|}$ and otherwise $L_v(x)=v(x)$.  By Strong Pareto, $L_w \succ L_v$ iff $\sum_{x \in Z} w(x)> \sum_{x \in Z} v(x)$, and $L_v \succ L_w$ iff $\sum_{x \in Z} w(x)> \sum_{x \in Z} v(x)$, while (by reflexivity) $L_w \sim L_v$ iff $\sum_{x \in Z} w(x)= \sum_{x \in Z} v(x)$ (since then they are the same distribution). 

By Permutation Invariance, if $w \succeq v$ then $w_i \succeq v_i$. By Convex Dominance, if $w \succeq v$, then $L_w \succeq  L_v$. So if $w \succeq v$, then $\sum_{x \in Z} w(x)\geq \sum_{x \in Z} v(x)$ as required. 

\end{proof}

\subsubsection{Second Step: The Intermediate Value Lemma}

The next Lemma contains the key idea of the theorem.

\begin{lemma}[Intermediate Value Lemma] Let $W= W_F$. Suppose $\succeq$ satisfies Axioms 1-3, $w, v \in W$ and there is an infinite $A \subseteq X$ such that:
\begin{itemize}
\item There are some $h, l \in \mathbb R$ with $h>l$ and infinite disjoint subsets of $A$, $H, L$ such that for all $x \in H$ $w(x)=h$ and for all $x \in L$ $w(x)=L$
\item For some $k \in \mathbb R$, for all $x \in A$ $v(x)=k$.
\end{itemize}
For any $m\in (l, h)$ let $w_m$ be such that $x \in A$ $w_m(x)=m$, while for all $x \in X - A$ $w_m(x)=w(x)$. Then if $w \succeq v$ then $w_m \succeq v$, and if $w_m \preceq v$ then $w_m \preceq v$. \end{lemma}

\begin{proof}
We consider just the case where $w \succeq v$, since the case of $w \preceq v$ is exactly parallel.

Let $n$ be the smallest natural number such that $(h-l)/n+l<m$, and consider a partition of $L$ into $n-1$ infinite disjoint sets, which we label $L_1\dots L_{n-1}$. For all $i$ with $1\leq i \leq n-1$ let $\pi_i$ be a permutation which induces a bijection between $L_i$ and $H$, while for all $x \in X- (H \cup L_i)$, $\pi_i (x)=0$. (The existence of such a permutation is guaranteed by the Axiom of Choice.) Let $\pi_n$ be the identity permutation. For all $i$ with $1 \leq i \leq n$, let $w_i$ be $\pi_i(w)$. Note that for all such $i$, $\pi_i(v)=v$. By Permutation Invariance, and this fact, for all $1 \leq i \leq n$, $w_i \succeq v$. By Convex Dominance, with $\alpha_i = 1/n$, $\sum_{1\leq i\leq n} \alpha_i w_i \succeq v$. For all $x \in A$, $(\sum_{1\leq i\leq n} \alpha_i w_i)(x)=(h-l)/n+l$, which by assumption is less than $m$ (for all $x \in X-A$, these two worlds take the same values). So, by Strong Pareto $w_m \succeq \sum_{1\leq i\leq n} \alpha_i w_i$. By Transitivity, $w_m \succeq v$ as required.
\end{proof}

\begin{corollary}[Split Value Corollary] Let $W= W_F$. Suppose $\succeq$ satisfies Axioms 1-3, $w, v \in W$ and there is an infinite $A \subseteq X$ such that:
\begin{itemize}
\item There are some $h, l \in \mathbb R$ with $h>l$ and infinite disjoint subsets of $A$, $H, L$ such that for all $x \in H$ $w(x)=h$ and for all $x \in L$ $w(x)=L$
\item For some $k \in (l, h)$, for all $x \in A$ $v(x)=k$.
\end{itemize}
If for all $x \in X-A$, $w(x)\leq v(x)$, then $w \not \succeq v$. If for all $x \in X-A$, $w(x) \geq v(x)$ then $v \not \succeq v$
\end{corollary}

\begin{proof} As above, for any $m\in (h,l)$, let $w_m$ be defined so that for all $x\in A$, $w_m(x)=m$, and for all $x \in X-A$, $w_m(x)=w(x)$. If $w \succeq v$ then by the Intermediate Value Lemma, for any $m \in (l,h)$ with $m<k$, $w_m \succeq v$. If for all $x \in X-A$ $w (x) \leq v(x)$, this contradicts Strong Pareto. Similarly, if $v \succeq w$, then for $m \in (l,h)$ with $m>k$, $v \succeq w_m$. Again if for all $x \in X-A$, $w(x) \geq v(x)$, this contradicts Strong Pareto. \end{proof}

\subsubsection{Third Step: Proof of the Theorem}

\begin{proof}
\emph{Case 1.} Suppose $w \sumord{} v$ but $v \not \sumord{} w$, i.e. that $\sum_{x \in X} w(x) -v(x)$ either converges unconditionally to a value $> 0$ or diverges to positive infinity. In either case, given that worlds are finite-valued, the set $A=\{ x | v(x)>w(x)\}$ must be finite, and there is a finite subset of $X$, $B$ such that $\sum_{x\in A \cup B} w(x)-v(x)  > 0$. Let $w^-$ be such that if $x \in A \cup B$ $w^-(x)=w(x)$, and for all other $x$, $w^-(x)=v(x)$. By Strong Pareto, $w \succeq w^-$. So if $v \succeq w$, we would have $v \succeq w^-$. But this is impossible by Lemma \ref{lemma1}, which requires (given that the worlds differ only on a finite set) that either $w^- \succ v$ or $w^- \perp v$.

\emph{Case 2.} Suppose that $w \not \sumord v$ and $v \not \sumord w$, so that there are infinitely many $x$ on which $w(x)>v(x)$ and also infinitely many $x$ on which $v(x)>w(x)$. We show that $w \not \succeq v$. Since this is proved without loss of generality, the claim follows.

Let $a$ be the largest value in the set $\{ r |$ for some $x \in X$ such that $w(x)>v(x)$, $w(x)=r\}$, and $b$ be the smallest value in the set $\{ r | $ for some $x \in X$ such that $w(x)>v(x)$ $v(x)=r\}$. (The existence of such an $a$ and $b$ is guaranteed by the fact that worlds are finite-valued.) Let $c$ be the largest value such that $\{x \in X | w(x)<v(x)$ and $w(x)=c\}$ is infinite and $d$ the smallest value such that $\{x \in X | w(x)=c$ and $v(x)=d \}$ is infinite. (Again the existence of such $c$ and $d$ is guaranteed by the fact that worlds are finite valued.) Finally let $A= \{x \in X | w(x)> v(x)\}$ and $E=\{ x \in X | w(x)=c $ and $v(x)=d\}$, and let $E_1$ $E_2$ be infinite disjoint subsets of $E$. (Note that $E \subset \{x \in X | v(x) >w(x)\}$.

By assumption $a>b$ and $d>c$. We consider three subcases governing inequalities between remaining values.

\underline{Case 2a: $a > c, d\geq b$.} Let $w^+$ be defined so that for all $x\in A$, $w^+(x)=a$ and for all other $x$ $w^+(x)=w(x)$. Let $v^-$ be defined so that for all $x \in A \cup E_1$, $v(x)=b$, while for all other $x$, $v^-(x)=v(x)$. By Strong Pareto, $w^+ \succeq w$ and $v \succeq v^-$, so if $w \succeq v$, then $w^+ \succeq v^-$. For some $m \in (c,d)$ let $w^+_m(x)=m$ if $x \in A \cup E_1$ while for all other $x$, $w^+_m(x)=w^+(x)$. By the Intermediate Value Lemma, if $w^+ \succeq v^-$, then $w^+_m \succeq v^-$. Now, for $x \in A \cup E_1 \cup E_2$, let $(w^+_m)^+(x)=m$, while for all other $x$, let $(w^+_m)^+(x)=w(x)$. The only difference between $w^+_m$ and $(w^+_m)^+$ is that the latter improves all $x \in E_2$ to take value $m$ instead of $c$. Since $m >c$, by Strong Pareto $(w^+_m)^+ \succeq w^+_m$, and so if $w \succeq v$ then $(w^+_m)^+ \succeq v^-$.  But the Split Value Corollary rules this out. To see that this corollary applies, note that for $x \in A \cup E_1$, $v^-(x)=b$, for $x \in E_2$, $v^-(x)=d$; for all $x \in A\cup E$ $(w^+_m)^+= m$ (with $m \in (b,d)$), and for all other $x$, we have $(w^+_m)^+(x)\leq v^-(x)$. So we cannot have $(w^+_m)^+ \succeq v^-$, which in turn means that $w \not \succeq v$.

\underline{Case 2b: $a >c, b>d$.} Let $w^+$ be defined so that for all $x$ such that $w(x)>v(x)$, $w^+(x)=a$ and for all other $x$ $w^+(x)=w(x)$ (as in the previous case). Let $v^-$ be defined differently (since $b>d$), so that for all $x$ such that $w(x)>v(x)$, $v^-(x)=d$. By Strong Pareto, $w^+ \succeq w$ and $v \succeq v^-$, so if $w \succeq v$, then $w^+ \succeq v^-$. But the Split Value Corollary rules this out. To see that this corollary applies note that $d \in (c,a)$, for all $x \in A\cup E$ $v^-(x)=d$, for all $x \in A$, $w^+(x)=a$, for all $x \in E$, $w^+(x)=c$, and for all $x \in X-(A\cup E)$, $v(x) \geq w(x)$. So  we cannot have $w^+ \succeq v^-$ which in turn means that $w \not \succeq v$. 

\underline{Case 2c: $a\leq c$.} By definition $a >b$ and $d >c$, so in this case we have $d>b$. Let $w^+$ be defined so that for all $x$ where $w(x)>v(x)$, $w^+(x)=c$ and for all other $x$ $w^+(x)=w(x)$. Let $v^-$ be defined so that for all $x$ where $w(x)>v(x)$, $v(x)=b$, while for all other $x$, $v^-(x)=v(x)$. As usual, Strong Pareto implies that $w ^+ \succeq w$ and $v \succeq v^-$ so that if $w \succeq v$, $w^+ \succeq v^-$. But again the Split Value Corollary rules this out. To see that the corollary applies note that $c \in (b,d)$, for all $x \in A\cup E$ $w^+(x)=c$, for all $x \in A$, $v^-(x)=b$, for all $x \in E$, $v^-(x)=d$, and for all $x \in X-(A\cup E)$, $v(x) \geq w(x)$. So  we cannot have $w^+ \succeq v^-$, which in turn means that $w \not \succeq v$.
\end{proof}

\subsection{Proof of Part 1 of Proposition \ref{nonmaximality}}

\begin{proof} \ \begin{itemize}
\item \emph{Reflexivity:} Immediate from the definition, given that $ \sumord{} $ satisfies it and is sufficient for the preorder.
\item \emph{Axiom 1 (Strong Pareto)} By the fact that $ \sumord{} $ satisfies Strong Pareto and that the comparisons introduced by (ii) can never occur when one world Pareto-dominates the other.
\item \emph{Axiom 2 (Permutation Invariance):} Immediate from the definition (differences between worlds taken individual by individual are invariant under permutations).
\item \emph{Transitivity:} We want to show that if $w \blacktriangleq v$ and $v \blacktriangleq u$ then $w \blacktriangleq u$. 
\begin{itemize}
\item Suppose first that $w \sumord{} v$. 
\begin{itemize} 
\item If $v \sumord{} u$, then $w \sumord{} u$, and hence $w \blacktriangleq u$ as required. 
\item So suppose that $v \blacktriangleq u$ because of condition (ii). If $w \sumord{} u$, then the claim is established. If not, then (ii)(a) is satisfied for $w$ and $u$. To show that (ii)(b) is satisfied for $w$ and $u$, note first that if the $A$ which witnesses condition (ii)(b) for $v$ in comparison to $u$ has finite (or null) intersection with $\{x | v(x) > w(x) \}$ this is obvious. If the intersection is infinite, note that $\sum_{x \in \{x | v(x)>w(x)\} } v(x)-w(x)$ must converge unconditionally, so for any $k>0$ there are at most finitely many $x$ with $v(x)-w(x)>k$. In particular, for any $c$ which witnesses condition (ii)(b) for $v$ in comparison to $u$ there are at most finitely many greater than $c-\epsilon$ for any $\epsilon>0$, $\epsilon<c$. So, there is an infinite $Z \subseteq A$ such that for all $x \in Z$ $v(x)-w(x)<c-\epsilon$ and hence (given that for all $x \in A$ $v(x)-u(x)>c$), such that for all $x \in Z$, $w(x)-u(x)>\epsilon$. So $Z$ and $\epsilon$ witness (ii)(b) for $w$ and $u$. Similar reasoning establishes (ii)(c): if $\{x | v(x)>w(x)\}$ is finite, or infinite but with an absolutely convergent sum, then since there is no $d>0$ and infinite $B$ such that $u(x)-v(x)>d$, there is also no $d>0$ and infinite $B$ such that $u(x)-w(x)>d$.
\end{itemize}
\item Suppose now that $w \blacktriangleq v$ because of condition (ii), and that $v \blacktriangleq u$. \begin{itemize}
\item If $v \sumord{} u$, then either $w \sumord{} u$, in which case the claim is shown, or (ii)(a) is satisfied for the comparison of $w$ and $u$. By similar arguments to those in the previous case, (ii)(b) and (ii)(c) must also be satisfied for $w$ and $u$ given that they were satisfied for $w$ and $v$.
\item If $v \blacktriangleq u$ by condition (ii), then either $w \sumord{} u$, in which case the claim is shown, or (ii)(a) is satisfied for the comparison of $w$ and $u$. Clearly (ii)(b) must also be satisfied, given that it is satisfied both in the comparison of $w$ and $v$, and in the comparison of $v$ and $u$. And (ii)(c) will also be satisfied: given that there is no $d>0$ and infinite $B$ satisfying the condition for either the comparison of $w$ and $v$ or the comparison of $v$ and $u$, there also can't be one for the comparison of $w$ and $u$.: So $w \blacktriangleq u$, as required.
\end{itemize}
\end{itemize}
\item \emph{Axiom 3 (Quasi-Independence)} We want to show that if $w \blacktriangleq v$, then for any $\alpha \in (0,1)$ and $u \in \mathbb R^X$ $\alpha w + (1-\alpha)u \succeq \alpha v + (1-\alpha)u$. If $w \sumord{} v$, the claim is obvious. If $w$ satisfies condition (ii) above with respect to $v$ then $\alpha w$ will satisfy (ii) with respect to $\alpha v$. Since the property is defined in terms of differences between worlds it is preserved by addition of a constant $(1-\alpha) u$ to each world. 
\end{itemize}
\end{proof}

\end{appendices}

\begin{small}
\begin{singlespacing}

\bibliographystyle{plainnat}
\bibliography{references}
\end{singlespacing}

\end{small}

\end{document}